\documentclass[submission,copyright,creativecommons]{eptcs}
\providecommand{\event}{Ninth International Symposium on Games, Automata, \\Logics and Formal Verification (GANDALF 2018)} 
\usepackage{breakurl}             
\usepackage{amsmath}
\usepackage{amsfonts}
\usepackage{amsthm}
\usepackage{amssymb}
\usepackage{algorithm}
\usepackage{algorithmicx}
\usepackage{algpseudocode}
\usepackage[english]{babel}
\newtheorem{theorem}{Theorem}

\newtheorem{lemma}[theorem]{Lemma}
\newtheorem{definition}[theorem]{Definition}
\newtheorem{example}[theorem]{Example}

\title{Parameterized Verification of Coverability in Well-Structured Broadcast Networks}
\author{A.R.Balasubramanian
\institute{Chennai Mathematical Institute\\
Chennai, India}
}
\def\titlerunning{Coverability of Well-Structured Broadcast Networks}
\def\authorrunning{A.R.Balasubramanian}
\begin{document}
\maketitle

\begin{abstract}
Parameterized verification of coverability in broadcast networks with finite state processes has been studied for different types of models and topologies. In this paper, we attempt to develop a theory of broadcast networks in which the processes can be well-structured transition systems. The resulting formalism is called \textit{well-structured broadcast networks}. 
We give an algorithm to decide coverability of well-structured broadcast networks when reconfiguration of links between nodes is allowed. Further, for various types of communication topologies, we also prove the decidability of coverability in the static case as well. We do this by showing that for these types of static communication topologies, the broadcast network itself is a \textit{well-structured transition system}, hence proving the decidability of coverability in the broadcast network.
\end{abstract}

\section{Introduction}

Specification and verification of infinite-state systems is a challenging task. Over the last two decades, various techniques have been proposed for checking safety and other properties of such systems, with one of the most prominent among them being the concept of a \textit{well-structured transition system} \cite{wsts,wqts}. A well-structured transition system is a transition system equipped with a well-quasi ordering on its states. Under some mild assumptions on the transition system, it is known that coverability in such systems is decidable.

Parameterized verification comprises of studying networks formed of anonymous agents executing the same code which interact with each other through some medium of communication, like broadcast, rendez-vous and shared variables \cite{crowd,reason,stochasticreg}. Its aim is to certify the correctness of all instances of the model, independently of the (parameterized) number of agents. Such problems are usually phrased in terms of infinite-state systems, to which technqiues from infinite-state verification theory can be applied. Indeed, a lot of results on parameterized verification prove that the underlying infinite state space of networks is a well-structured transition system. \cite{static, clique,register,time}

Broadcast networks are a formalism introduced in \cite{static}, in which the agents can broadcast messages simultaneously to all its neighbors. The number of agents and the communication topology are fixed before the start of the execution. Parameterized verification of such systems involves checking whether a specification holds irrespective of the number of agents or the communication topology. One of the prominent specifications considered in literature for such systems is the problem of coverability: \textit{does there exist an initial configuration from which at least one agent may reach a particular state}. In \cite{static}, the authors prove that the coverability problem for broadcast networks is undecidable even when the agents are finite state processes. Also, undecidability  has been proven for broadcast networks restricted to bounded-diameter topologies \cite{clique} and decidability has been proven for bounded-path topologies \cite{static},
bounded-diameter and degree topologies \cite{clique}, and clique topologies. Further, when we allow \textit{reconfigurations} of links in the underlying communication topology, there exists a polynomial time algorithm to decide coverability of broadcast networks comprising of finite-state processes \cite{mobile}.
This result perhaps seems surprising, since the reconfigurable case looks like a generalization of the static case. We note that a similar dichotomy exists between the verification of perfect and lossy channel systems. There has also been some work in extending the results of parameterized verification from the finite-state case to probabilistic automata \cite{prob,probtime} and timed automata \cite{time}. With the theory of broadcast networks having been explored for these various types of models, it seems natural to try to develop a theory of broadcast networks with well-structured transition systems as the underlying processes.

In this paper we study the coverability problem for broadcast networks where each process can be a \textit{labeled well-structured transition system}.
In such systems, the underlying process itself can have infinite states. We call such systems \textit{well-structured broadcast networks}. We prove that the coverability problem is decidable for various classes of restricted topologies in this setting. In particular, we prove decidability for the set of all clique topologies, the set of all path-bounded topologies and the set of all topologies with bounded diameter and degree. We show that for these sets of topologies with well-structured transition systems as processes, the underlying state space of networks is itself a well-structured transition system. 
We also give an algorithm for deciding the coverability of a configuration for well-structured broadcast networks when reconfiguration of edges is permitted between the interacting agents.\\

\textbf{Acknowledgements: } I am extremely grateful to Nathalie Bertrand and Nicolas Markey for useful discussions on the topic and also for assisting in the preparation of this paper. I would also like to thank Thejaswini K.S and Mirza Ahad Baig for comments on early drafts of this paper and the anonymous reviewers for their valuable feedback.

\section{Well-structured broadcast networks}

In this section, we recall results about \textit{well-structured transition systems} \cite{wsts,wqts} and use them to 
define \textit{well-structured broadcast networks}. We also introduce the \textit{reconfiguration} semantics for such networks as a way
of modelling link changes that might occur in the underlying communication topology.

\subsection{Well-structured transition systems}

\begin{definition}
A well-quasi ordering (wqo) $\le$ on a set $X$ is a reflexive, transitive binary relation s.t. any infinite sequence of elements $x_0,x_1,\cdots$ contains an increasing pair $x_i \le x_j$ with $i < j$.
\end{definition}

\begin{definition}
A labelled well-structured transition system (labelled WSTS) is a tuple $TS = (S,\Sigma,S_0,R,\le)$ where 

\begin{itemize}
\item $S$ is a set of configurations
\item $\Sigma$ is a finite set of symbols called the alphabet
\item $R \subseteq S \times \Sigma \times S$ is the transition relation
\item $S_0$ is the set of initial configurations
\item $\le \subseteq S \times S$ is a well-quasi order between states such that: 
	\begin{itemize}
	\item $\le$ is compatible with $R$, i.e., if $a \in \Sigma$ and $s_1 \le t_1$ and $(s_1,a,s_2) \in R$, 
	then $\exists \; t_2$ s.t. $(t_1,a,t_2) \in R$ and $s_2 \le t_2$
	\end{itemize}
\end{itemize}

\end{definition}

Note that our definition of labelled WSTS is robust in the sense that if we restrict the WSTS to transitions of a particular label, we still get a WSTS. A WSTS is called \textit{finitely branching} if 
for each $s \in S$, there are only finitely many
transitions of the form $(s,a,s') \in R$. We will consider only finitely branching WSTS in this paper.

We call a set of configurations $I \subseteq S$, \textit{upward-closed} if $x \in I$ and $y \ge x$ implies $y \in I$. To any subset $I \subseteq S$,
we define $\uparrow I = \{x: \exists \; y \in I, \;
x \ge y\}$. In particular a set $I$ is upward-closed iff $I = \uparrow I$. 
A \textit{basis} for an upward-closed set $I$, is a set $I^b$ s.t. $I = \uparrow I^b$. 
It is known that for a wqo, every upward-closed set has a finite basis.

Given a set of configurations $I$, denote by $pre(I)$ the set $\{s' \in S: (s',a,s) \in R, \text{ for some } a \in \Sigma, s \in I\}$. For $i > 0$, let $pre^i(I):= \{s' \in S: (s',a,s) \in R, \text{ for some } a \in \Sigma, s \in pre^{i-1}(I)\}$ and let $pre^*(I):= \cup_{i \in \mathbb{N}} \; pre^i(I)$. We will write $s \rightarrow s'$ to mean that $s \in pre(s')$
and $s \xrightarrow{*} s'$ to mean that $s \in pre^*(s')$. A labelled WSTS is said to have
\textit{effective pre-basis} if given a finite basis for the upward-closed set $I$, we can compute a basis for the set $pre(I)$.

The \textit{coverability} problem for labelled WSTS is the following: Given a configuration $s$, decide if $\exists \; s',s_0$ s.t. $s_0 \in S_0$ and $s' \ge s$ and $s_0 \xrightarrow{*} s'$.

From \cite{wqts,wsts} it is known that 

\begin{theorem}
Coverability is decidable for labelled WSTS with effective pre-basis and a decidable wqo.
\end{theorem}

The idea behind the proof is as follows: Given a configuration $s$, we compute the following sequence of upward-closed sets $U_0 = \uparrow s$ and 
$U_{i+1} = pre(U_i)$. This sequence will eventually saturate to some $U_m$ which will give us a finite basis for $pre^*(U_0)$. Checking whether $s$ can be covered now amounts to checking if there is at least one initial configuraion in $\uparrow U_m$. 

Common examples of labelled WSTS include: Any finite state system, vector addition systems with states (VASS), Petri nets with reset arcs, Petri nets with transfer arcs and lossy counter machines.\\

A labelled WSTS might be an infinite state system and so it is infeasible to describe the entire set of configurations in an explicit way.
Usually, a labelled WSTS $TS = (S,\Sigma,S_0,R,\le)$ is given by means of a finite description $P = (Q,\Sigma,Q_0,\Delta,\cdots)$. The 
finite description may have additional structure like counters, causal relations etc. 
The structure of the transition relation $\Delta$ depends on the type of labelled WSTS that it describes.

\begin{example}
Let $(Q,\Sigma,Q_0,\Delta,V)$ be a vector addition system with states (VASS) where $Q$ is a finite set of states, $\Sigma$ is a finite alphabet, $Q_0$ is a set of initial states, $V$ is a finite set of vectors over $\mathbb{Z}^d$ 
(for some $d$) and $\Delta$ is of the form $\Delta \subseteq Q \times \Sigma \times V \times Q$. This describes a labelled WSTS $(S,\Sigma,S_0,R)$ where
$S$ is the set of all configurations, i.e., $S = \{(p,u): (p,u) \in Q \times \mathbb{N}^d\}$, $S_0 = \{(p,u): (p,u) \in Q_0 \times 0^d\}$. The transition
relation $R$ is defined in the following manner: 
$((p,u),a,(q,w)) \in R$ iff $\exists \; v \in V,
(p,a,v,q) \in \Delta$ s.t. $u+v \ge 0$ and $w = u+v$. In this case we see that each transition $((p,u),a,(q,w)) \in R$ is 
described by a transition $(p,a,v,q) \in \Delta$.
\end{example}

For the rest of this paper we will assume that every labelled WSTS $TS$ will be given by means of a finite description $P$.

\subsection{Well-structured broadcast networks}

In this section, we define \textit{well-structured broadcast networks} and also introduce the \textit{reconfiguration} semantics.

Throughout the paper, we fix a finite alphabet $\Sigma$.
Let the set of symbols $\{!!a: a \in \Sigma\}$ be 
denoted by $\Sigma_b$ and let the set of symbols
$\{??a: a \in \Sigma\}$ be denoted by $\Sigma_r$.

\begin{definition}
A \textbf{process} is a labelled well-structured transition system $TS = (S,\Sigma_b \cup \Sigma_r,S_0,R,\le)$.
\end{definition}

A well-structured broadcast network consists of several copies of a single process $TS$. Each configuration of such a network is 
an undirected graph in which each node is labelled by a configuration $s \in S$. Intuitively, the labels $!!a$ and $??a$ correspond to
broadcasting and receiving messages according to the topology specified by the underlying graph. Formally,

\begin{definition}
An $S$-graph is a graph $G = (V,E,L)$ where $L$ is a labelling function $L: V \to S$.
\end{definition}

An $S$-graph represents an undirected graph in which each node $v \in V$ is executing the same process $TS$ and is currently in the
configuration $L(v)$.

We now use the notion of a process to define a transition system called the \textit{well-structured broadcast network}.

\begin{definition}

Given a process $TS = (S,\Sigma_b \cup \Sigma_r,S_0,R,\le)$, a \textit{well-structured broadcast network} is a tuple $BN(TS) = (\Theta,\Theta_0,\rightarrow)$, where 

\begin{itemize}
\item $\Theta$ is the set of all finite $S$-graphs
\item $\Theta_0$ is the set of all finite $S_0$-graphs and 
\item $\rightarrow$ is defined as follows: If $\theta = (V,E,L)$ and $\theta' = (V,E',L')$, then $\theta \xrightarrow{a} \theta'$ if 

\textbf{Broadcast: } $E = E'$ and $\exists \; v \in V$ s.t. 
\begin{enumerate}
\item $(L(v),!!a,L'(v)) \in R$
\item $(L(u),??a,L'(u)) \in R$ for every node $u$ connected to $v$
\item $L'(w) = L(w)$ for every other node $w$
\end{enumerate}

\end{itemize}

\end{definition}

Whenever the process $TS$ is clear from the context, we refer to the broadcast network only by $BN$.

The well-structured broadcast network can be thought of as follows: 
We have a graph in which each vertex runs a copy of the process $TS$ and the current 
label of the vertex $v$ denotes the configuration of the process at $v$.
At each time step, a process in some vertex $v$ chooses to broadcast a message $!!a$
and it is received ($??a$) by all its neighbors $u$.

Notice that this formulation of broadcast networks does not permit changes in links in the underlying topology.
To model such changes, we use the notion of \textit{reconfigurations}. A 
\textit{reconfigurable well-structured broadcast network} is a well-structured broadcast network in which along with \textit{Broadcast} moves, we 
also allow transitions of the following kind: $\theta = (V,E,L) \rightarrow \theta' = (V,E',L')$ if \\

\textbf{Reconfiguration: } $L = L'$ and $E' \subseteq V \times V \setminus \{(v,v): v \in V\}$\\

Any reconfiguration corresponds to a non-deterministic change in the underlying network topology of the processes. We denote the resulting transition system by $RBN(TS)$.

Given a well-structured broadcast network $BN$, the \textit{coverability} problem, given a configuration $s$, is to decide if there exists an initial graph such that by a series of transitions, we can reach a network topology in which at least one agent attains a configuration $s'$ which covers $s$. More formally, we consider the following problem: Given a configuration $s \in TS$, decide if there
exist $\theta \in \Theta, \theta_0 \in \Theta_0$ and $s' \ge s$ s.t. $\theta_0 \xrightarrow{*} \theta$ and $s'$ is the label of some process in $\theta$.
Notice that this is not the same as asking if $s$ is coverable in $TS$.

\begin{example}
Consider a finite automaton with just two states $q,q'$ s.t. $q \xrightarrow{??a} q'$. This trivially describes a labelled WSTS and coverability in this
case is just reachability. Notice that the state $q'$ can never be reached in the $BN$ described by this automaton. But when we treat this just as a labelled 
transition system without the broadcast network semantics, it is clear that $q'$ can be reached from $q$ in the transition system $TS$. 
To distinguish this, we refer to these two cases distinctly as coverability in $BN$ and coverability in $TS$.
\end{example}

It is known that the coverability problem for well-structured broadcast networks $BN(TS)$ is undecidable. Indeed, it is undecidable even when 
$TS$ is a finite state transition system \cite{static}.
As a first step for overcoming the undecidability, we look at $RBN(TS)$, i.e., the reconfigurable well-structured broadcast network associated with a process $TS$.

\section{Coverability problem for reconfiguration semantics}

In this section, we prove that the coverability problem for the reconfiguration semantics is decidable. In particular, 
we present an algorithm which when given a process $TS$ and a configuration $s$, returns true iff the configuration $s$ can be covered in $RBN(TS)$.

Let $P = (Q,\Sigma_b \cup \Sigma_r,Q_0,\Delta,\cdots)$ be a finite specification of the process $TS$. We assume that for each transition $t \in \Delta$, 
we are able to compute all minimal configurations $c_t(1),c_t(2),\cdots,c_t(l) \in TS$ 
s.t. $t$ is enabled at $c_t(i)$ for each $i$ in the transition system $TS$. (This set is always finite, since the underlying order is a wqo).
Let $c_t = \{c_t(1),\cdots,c_t(l)\}$. Notice that this computation concerns only the semantics of the transition system $TS$ and not that of $RBN$.
For many systems such a computation will be fairly straightforward.

\begin{example}
\begin{enumerate}
\item If $P$ describes a finite state system, then for any transition $t$ we can search through the space of all configurations in 
$TS$ and compute the ones at which $t$ is enabled.

\item In a VASS, given a transition $t = (p,a,v,q)$, it is clear that the configuration $(p,u)$ where $u_i = max(0,-v_i)$ 
is the minimal configuration at which $t$ is enabled. 

\item Since transfer arcs and reset arcs in Petri nets have the same precondition as normal transitions, it follows that computation of 
minimal configurations in these cases is similar to VASS \cite{resetarcs}.
\end{enumerate}
\end{example}

A broadcast transition is a transition in $\Delta$ labelled by letters from $\Sigma_b$. Similarly, a receive transition is one labelled by letters from $\Sigma_r$. Further for each letter $a \in \Sigma$, we define $B_a$ to be the set of all broadcast transitions labelled by $!!a$. Similarly, 
we define $R_a$ to be the set of all receive transitions labelled by $??a$. Let $Rec = \cup_{a} R_a$. In the sequel, given some finite specification $P'$, we denote by $TS(P')$ the transition system that $P'$ describes and we denote by $RBN(P')$ the reconfigurable broadcast network that $TS(P')$ describes.

The coverability algorithm for $RBN$ is given in Algorithm 1. The algorithm proceeds as follows: As a first step, from the original process $P$ we remove all transitions in $Rec$, to get a modified process $P'$. At each iteration of the main loop, for each letter $a$ and for each transition $t$ in $B_a$, we check if atleast one configuration from the set $c_t$ can be covered in the current process $P'$.
Intuitively, this means that some agent in the network can reach a configuration, from which it would be capable of broadcasting the letter $a$. At this point, we update the process $P'$ by adding all the receive transitions labelled by $a$. Whenever in the future, an agent wants to take a transition labelled by $??a$, it can do so now, because we can make another agent reach a configuration capable of broadcasting $a$, and then reconfigure the network, so that both these agents share an edge. This is where the reconfiguration semantics of the network plays a prominent role in checking the coverability of a configuration.
We keep doing this until no more transitions can be added, at which point we check if the required configuration is coverable in the resulting process obtained.\\

\begin{algorithm}[h]
\caption{Coverability algorithm for reconfiguration semantics}
\begin{algorithmic}[1]
\State \textbf{Input: } A finite specification $P = (Q,\Sigma_b \cup \Sigma_r,Q_0,\Delta,\cdots)$ and a configuration $s \in TS(P)$
\State \textbf{Output: } Whether $s$ is coverable in the transition system $RBN$\\

\State $P':= P[ \Delta \leftarrow \Delta \setminus Rec]$ \Comment{Remove all receive transitions from $P$ to get $P'$}
\State $SubAlp:= \Sigma$
\Repeat
	\State $AddT:= \emptyset$ 
	\ForAll{$a \in SubAlp$ } \Comment{Look for symbols that can be broadcast in $TS(P')$}
		\ForAll{$t \in B_a$ }
			\If{$\exists i$ s.t. $c_t(i)$ is coverable in $TS(P')$}
				\State $AddT:= AddT \cup R_a$ \Comment{And store the receive transitions in $AddT$}
				\State $SubAlp:= SubAlp \setminus \{a\}$
			\EndIf
		\EndFor
	\EndFor\\
		
	\State $P':= P'[\Delta \leftarrow \Delta \cup AddT$] \Comment{Add all transitions from the set $AddT$ to $P'$}

\Until{$AddT = \emptyset$}\\

\If{$s$ is coverable in $TS(P')$}
	\State \textbf{return true}
\Else
	\State \textbf{return false}
\EndIf

\end{algorithmic}
\end{algorithm}

Notice that at any point in the algorithm, the transition system $TS(P')$ will always be a well-structured transition system. Indeed at the beginning of the code, $TS(P')$ is exactly the transition system obtained by removing all transitions labelled by $\Sigma_r$ from $TS(P)$ and since $TS(P)$ was a WSTS, $TS(P')$ will also remain a WSTS. Similarly, at each update of the $TS(P')$, we add all transitions of the form $??a$ for some symbol $a \in \Sigma$. Hence, the new transition system $TS(P')$ continues to be a WSTS.

The coverability tests in lines 10 and 20 refer to coverability in the transition system $TS(P')$. 
Also notice that whenever the algorithm increases the
cardinality of the set $AddT$, it decreases the size of $SubAlp$ by 1. Since, the transitions added to $AddT$ are labelled by symbols from $SubAlp$ and since $SubAlp$ is finite, it follows that eventually we can add no more transitions to $AddT$. Therefore, line 18 of the algorithm will eventually become true and so the algorithm always terminates.

Let $AddT_0 = \cup_{a \in \Sigma} B_a$ and for $i > 0$, let $AddT_i$ be the contents of the set $AddT$ at the end of the $i^{th}$ iteration of the $while$ loop. Further, let $P'_0 = P[\Delta \leftarrow \Delta \setminus Rec]$ and $P'_i = P'_{i-1}[\Delta \leftarrow 
\Delta \cup AddT_i]$ for $i > 0$, i.e., $P'_i$ denotes the description $P'$ obtained at the end of the $i^{th}$ iteration of the $while$ loop. Let the total number of iterations of the $while$ loop be $w$. Hence we have a sequence of processes $P'_0, P'_1, \cdots, P'_w$.

In the sequel, we will use the notation $s \xrightarrow{t} s'$ to denote the fact that the transition $t$ is enabled at $s \in S$ and $s'$ is the corresponding configuration reached upon executing $t$ from $s$.

The correctness of this algorithm follows by a series of lemmas.

\begin{lemma}
If a configuration $s$ is reachable in $TS(P'_i)$ for some $i$, then $s$ can be reached in the original reconfigurable broadcast network $RBN(P)$.
\end{lemma}

\begin{proof}

Let $s$ be a configuration which is reachable in the transition system $TS(P'_i)$.
Further wlog, let $i$ be the first index s.t. $s$ is reachable in $TS(P'_i)$. We will prove by induction on $i$ that the configuration $s$ is reachable in the broadcast network $RBN(P)$ as well.

Suppose $i = 0$. Since $s$ is reachable in $TS(P'_0)$, there exists a path $LP = s_0 \xrightarrow{t_0} s_1 \xrightarrow{t_1} s_2 \cdots
s_{n-1} \xrightarrow{t_n} s_n = s$ in the transition system $TS(P'_0)$. We prove the claim for $i = 0$ by a second induction on $n$. For the base case of $n = 0$, it is clear that $s_0$ is an initial configuration and so $s_0$ can be trivially reached in $RBN(P)$. Suppose $n > 0$. By our second induction hypothesis, the configuration $s_{n-1}$ is reachable in $RBN(P)$, i.e., there exist a reachable graph $\theta$ and a node $v \in \theta$ s.t.
$L(v) = s_{n-1}$. Since $LP$ is a path in $TS(P'_0)$,
the transition $t_n$ has to be a broadcast transition labelled by some letter $!!a$. 
Hence the node $v$ can broadcast $!!a$ and move into the configuration $s_n = s$.

Suppose $i > 0$. Again since $s$ is reachable in 
$TS(P'_i)$, there exists a path $LP = s_0 \xrightarrow{t_0} s_1 \xrightarrow{t_1} s_2 \cdots
s_{n-1} \xrightarrow{t_n} s_n = s$ in $TS(P'_i)$. We prove the claim by a second induction on $n$. For the base case of $n = 0$, again it is clear that $s_0$ is an initial configuration and so it is reachable in $RBN(P)$. Suppose $n > 0$. Similar to the above argument, by our second induction hypothesis, there exists a path in $RBN(P)$ of the form $\theta_0 \rightarrow \theta_1 \cdots \rightarrow \theta_m$ and a node $v \in \theta_m$ s.t. $L(v) = s_{n-1}$. We now consider two cases: Suppose $t_n$ is a broadcast 
transition labelled by $!!a$. It is then clear that $v$ can broadcast $!!a$ to reach the configuration $s_n = s$. 

Otherwise, $t_n$ is a receive transition labelled by some letter $??a$. Since $t_n \in TS(P'_i)$ it must have been added to the set $AddT_j$ for some $j \le i$. But notice that we add a new receive transition labelled by $??a$ in the $j^{th}$ iteration iff there exists a transition $t \in B_a$ and a minimal configuration $c_t(k)$ s.t.
$c_t(k)$ is coverable in the transition system $TS(P'_{i-1})$. Therefore, by definition of coverability $\exists \; s' \ge c_t(k)$ s.t. $s'$ is reachable in $TS(P'_{i-1})$. By our primary induction hypothesis, $s'$ is reachable in $RBN(P)$. So let $\theta'_0 \rightarrow \theta'_1 \cdots \rightarrow \theta'_l$ be a path in $RBN(P)$ and let $v' \in \theta'_l$ s.t. $L(v') = s'$. Notice that by the property of compatibility, there is a broadcast transition labelled by $!!a$ which is enabled at $s'$.

Now consider the initial graphs $\theta_0$ and $\theta_0'$. Execute the first run from $\theta_0$ so that it reaches the graph configuration $\theta_m$. Now, execute the second run from the initial graph $\theta_0'$ so that it reaches the graph configuration $\theta'_l$. This can be done since these two executions are independent of each other. Now add a link between $v$ and $v'$ and broadcast the message $!!a$ from $v'$. Hence $v$ will receive the message $??a$ and will move into the configuration $s_n = s$. 
\end{proof}

\begin{lemma}
If $s$ is reachable in the reconfigurable broadcast network $RBN(P)$, then $s$ is reachable in $TS(P'_w)$.
\end{lemma}

\begin{proof}
Suppose $s$ is reachable in the reconfigurable broadcast network $RBN(P)$. Therefore, there exists an intital path $LP = \theta_0 \rightarrow \theta_1 
\rightarrow \cdots \theta_{n-1} \rightarrow \theta_n$ and a node $v$ s.t. $L_{\theta_n}(v) = s$. We will prove the claim by induction on $n$. The claim is clear for the base case of $n = 0$.

Suppose $n > 0$. Consider the configuration of $v$ in the graph $\theta_{n-1}$, i.e., the configuration $L_{\theta_{n-1}}(v)$ and let it be denoted by $s'$. If $s' = s$, then by the induction hypothesis we are done. Suppose $s' \neq s$. Therefore, there should be a transition from $s'$ to $s$, i.e., $\exists \; t$ s.t. $s' \xrightarrow{t} s$. We now have two cases:

\begin{itemize}
\item Suppose $t$ is a broadcast transition labelled by $!!a$. By induction hypothesis, $s'$ is reachable in $TS(P'_w)$. Since all broadcast transitions are present in $P'_w$, it follows that $s$ is reachable in $TS(P'_w)$ as well.

\item Suppose $t$ is a receive transition labelled by $??a$. Hence the node $v$ in $\theta_{n-1}$ received a message $??a$ and so there should have been a node $u \in \theta_{n-1}$ in configuration $s_u'$ s.t. $u$ broadcast a message $!!a$ to reach some configuration
$s_u$ in the graph $\theta_n$. By induction hypothesis, $s_u'$ is reachable in the transition system $TS(P'_w)$. Hence there exists at least one transition with the broadcast label $!!a$ which is enabled in $TS(P'_w)$. This means that there exists at least one minimal configuration $c \in TS(P'_w)$ s.t. a broadcast transition labelled by $!!a$ is enabled at $c$ and $c$ is coverable. Hence for the letter $a$, line $10$ of the algorithm will eventually become true and so the transition $t$ would have been added to $P'_w$. This means that the transition $s' \xrightarrow{t} s$ is present in $TS(P'_w$). By induction hypothesis, $s'$ is reachable in $TS(P'_w)$ and so $s$ is reachable as well.
\end{itemize}
\end{proof}

Hence, we have

\begin{theorem}
Coverability in reconfigurable well-structured broadcast networks is decidable.
\end{theorem}

\begin{proof}
Notice that the algorithm returns its answer based on whether the given configuration $s$ is coverable in $TS(P'_w)$ or not. 

Suppose $s$ is coverable in $TS(P'_w)$. Therefore, 
$\exists \; s' \ge s$ s.t. $s'$ is reachable in $TS(P'_w)$ and so by Lemma 10, $s'$ is reachable in $RBN(P)$. Therefore $s$ is coverable in $RBN(P)$. The other side of the proof follows by a similar argument involving Lemma 11.
\end{proof}

Notice that the main bottleneck in the running time of this algorithm are the coverability tests to the transition system $TS(P')$.\\

Finally, introducing \textit{arbitrary reconfigurations} in the model might not seem too realistic. But in fact, w.r.t \textit{coverability}, this model is equivalent to:

\begin{enumerate}
\item Static topology with intermittent nodes, i.e., a topology in which there are no reconfigurations but nodes can crash and 
restart in the \textbf{same} control state in which it crashed. \cite{errors}
\item Static topology with message loss, i.e., a topology in which there are no reconfigurations but messages may get lost arbitrarily. 
\cite{errors}
\item Asynchronous broadcast network with a bag model. \cite{asynch}
\item Asynchronous broadcast network with a lossy FIFO queue. \cite{asynch}
\item Globally constrained runs, i.e., a run in which the number of reconfigurations allowed in between two broadcasts can be atmost $k \ge 1$.
\cite{constraint}
\item Locally constrained runs, i.e., a run in which the number of reconfigurations \textbf{each node} is allowed to make in between two 
broadcasts can be atmost $k \ge 1$.
\cite{constraint}
\end{enumerate}

The proofs given in these papers are for the case when the processes are finite state systems. But these claims can be proved for the infinite state case as well, by noticing that the corresponding proofs go through even in the case of infinite state systems. Intuitively, this is because the equivalence proofs only manipulate the graph topology of the underlying model.

\section{Coverability problem for restricted topologies}

We have mentioned that the coverability problem in general for well-structured broadcast networks is undecidable \cite{static}. In the previous 
section, we gave an algorithm to decide coverability of \textit{reconfigurable well-structured broadcast networks}.
In this section, we investigate coverability in well-structured broadcast networks without reconfigurations, but the set of all underlying graphs that we will consider will be restricted.
In particular, we prove decidability results for three different classes of restricted topologies, namely bounded path topologies, clique topologies and
bounded diameter and degree topologies. All these results could be seen as extensions of results that have been proved for 
finite state processes \cite{clique,static}.

As a first step, we define the \textit{induced subgraph ordering} between two configurations which will be used extensively to prove 
decidability in all three classes of topologies:

\begin{definition}
Given two configurations $\theta_1 = (V_1,E_1,L_1), \theta_2 = (V_2,E_2,L_2) \in \Theta$, define $\theta_1 \sqsubseteq \theta_2$ iff there exists an injection $h: V_1 \to 
V_2$ s.t. $\forall u,v \in V_1$,

\begin{itemize}
\item $(u,v) \in E_1 \iff (h(u),h(v)) \in E_2$
\item $L_1(u) \le L_2(h(u))$
\end{itemize}
\end{definition}

In other words, the injection $h$ should preserve edges among vertices and also the order of their labeles w.r.t the well-quasi ordering.

\subsection{Bounded path topologies}

In this section, we prove that the coverability problem becomes decidable when we restrict to path bounded graphs. 
We will assume throughtout that a number $k$ is fixed.

In the sequel, given a configuration $\theta$, we will denote its vertex set by $V(\theta)$.
Similarly, $E(\theta)$ and $L(\theta)$ will be used to denote the edge set and the label function of $\theta$ respectively.

\begin{definition} 
A graph $G$ is called $k$-path bounded if the longest simple path in $G$ has length atmost $k$.
\end{definition}

Notice that this is not the same as considering graphs of diameter $k$. 
(A distinction between the two is the clique graph, whose diameter is 1, but whose longest simple path is $n-1$).

Given a process $TS$, we can now define $k$-bounded path  broadcast networks by restricting the set of configurations in $BN(TS)$ to $k$-bounded path topologies, i.e
we define a new transition system $BN^k(TS) = (\Theta^k,\Theta_0^k,\rightarrow)$, where $\Theta^k$ and $\Theta_0^k$ consists only those configurations from
$\Theta$ and $\Theta_0$ which are $k$-path bounded. Notice that in this model, no reconfigurations are allowed between nodes.

We will employ the theory of well-structured transition systems to prove that the coverability problem for $k$-bounded path broadcast networks is decidable. More specifically,
as a first step, we prove the following lemma.

\begin{lemma}
The set of all $k$-path bounded configurations with the induced subgraph ordering is a well-quasi ordering.
\end{lemma}

\begin{proof}
Follows from Ding's theorem \cite{Ding}.
\end{proof}

As a next step, we prove that the induced subgraph ordering is \textit{compatible} with $BN^k(P)$.

\begin{lemma}
For every $\theta_1,\theta_2,\theta_1' \in \Theta^k$ s.t. $\theta_1 \xrightarrow{a} \theta_2$ and $\theta_1 \sqsubseteq \theta_1'$, there exits $\theta_2'
\in \Theta^k$ s.t. $\theta_1' \xrightarrow{a} \theta_2'$ and $\theta_2 \sqsubseteq \theta_2'$.
\end{lemma}

\begin{proof}
Let $v$ be the vertex in $\theta_1$ which broadcasts the message $!!a$ and let $u_1,\cdots,u_l$ be the neighbors of $v$ which receive the message
$??a$. Since $TS$ is well-structured and since $L(\theta_1)(v) \le L(\theta_1')(h(v))$, it follows that there exists a transition $t'$
labelled by $!!a$ which is enabled at $L(\theta_1')(h(v))$. Similarly, since $L(\theta_1)(u_i) \le L(\theta_1')(h(u_i))$, it follows that there
exist transitions $t'_i$ labelled by $??a$ which is enabled at $L(\theta_1')(h(u_i))$. Since $h$ is an injection it follows that 
each $h(u_i)$ is a neighbor of $h(v)$. Hence, we can broadcast the message $!!a$ from $h(v)$ and receive the message $??a$ at $h(u_1),\cdots,h(u_l)$
in the configuration $\theta_1'$. Call the resulting configuration $\theta_2'$. It is clear that the same injection $h: V(\theta_1) = V(\theta_2)
\to V(\theta_1') = V(\theta_2')$ satisfies the required conditions for an order between $\theta_1'$ and $\theta_2'$.
\end{proof}

As a final step, we prove that $BN^k$ has \textit{effective pre-basis}.

\begin{lemma}
If $S \subseteq BN^k$ is an upward closed set and has a finite basis, then we can effectively compute a finite basis for $pre(S)$.
\end{lemma}

\begin{proof}
Let $\mathbb{B} = \{\theta_1,\cdots,\theta_n\}$ be a finite basis of the upward-closed set $S$. For each $\theta_i$, we will construct a finite number of graphs,
whose overall union will be a basis for $pre(S)$.

Let $\theta = \theta_i = (V,E,L)$ and let $G = (V,E)$. For every node $v \in V$ and for every letter $a \in \Sigma_b \cup \Sigma_r$, we can compute a basis
for the upward closure of $pre(\uparrow L(v))$ in the transition system $TS$ restricted to transitions labelled only by $a$. Let this basis be denoted by $B_a^v$. Recall that for a transition $t \in \Delta$, $c_t$ denotes the set of minimal configurations in which $t$ is enabled.

Consider all $k$-path bounded graphs $H_1,H_2,\cdots,H_l$ s.t. $G$ is an induced subgraph of each $H_i$ and $|V(H_i)| = |V(G)| + 1$ 
, i.e., each $H_i$ has one more vertex than $G$ and contains $G$ as an induced subgraph. Using these $k$-path bounded graphs, we will compute new $k$-path bounded graphs which will form a basis for $pre(S)$. 

Consider the following process of creating new \textit{labelled graphs} from the graph $G = (V,E)$.

\begin{enumerate}
\item Choose a vertex $v \in V$ and a letter $a \in \Sigma$. Let $u_1,u_2,\cdots,u_p$ be the set of all neighbors of $v$ in $G$.
\item Choose a configuration $c_v$ from $B_{!!a}^v$ and configurations $c_{u_i}$ from $B_{??a}^{u_i}$ (if they are non-empty) respectively.
\item Construct the labelled graph $G_{before} = (V,E,L')$ as follows:
	\begin{enumerate}
	\item Label the vertex $v$ with $c_v$ and label each $u_i$ with $c_{u_i}$ respectively.
	\item Label the remaining vertices with the same labels that they had in $\theta$.
	\end{enumerate}
\item Choose a broadcast transition $t_v$ enabled at $c_v$ labelled by $!!a$. Let $c_v \xrightarrow{t_v} c_v'$
\item Choose a receive transition $t_{u_i}$ for each $u_i$ s.t. $t_{u_i}$ is enabled at $c_{u_i}$ and is labelled by $??a$. Let $c_{u_i} \xrightarrow{t_{u_i}} 
c_{u_i}'$
\item Construct the labelled graph $G_{after} = (V,E,L'')$ as follows:
	\begin{enumerate}
	\item Label the vertex $v$ with $c_v'$ and label each $u_i$ with $c_{u_i}'$ respectively.
	\item Label the remaining vertices with the same labels that they had in $\theta$.
	\end{enumerate}
\item If $G_{after}$ lies in the set $S$, add $G_{before}$ as a basis element for the set $pre(S)$.
\end{enumerate}

Now consider a similar process of creating new \textit{labelled graphs} from the graph $H_i$ for each $i$.

\begin{enumerate}
\item Fix an injection $h: G \to H_i$ and fix a letter $a \in \Sigma$.
\item Let $v$ be the vertex in $H_i$ which is not in the image of $G$, i.e., $v \notin h(G)$ and let $u_1,\cdots,u_p$ be the set of all neighbors of $v$ in $H_i$.
\item Choose a broadcast transition $t$ labelled by $!!a$ and choose a configuration $c_v$ from the set $c_t$. Let $c_v
\xrightarrow{t} c_v'$.
\item Choose a configuration $c_{u_j}$ from $B_{??a}^{u_j}$ for each $u_j$.
\item Construct the labelled graph $H_{before}$ from the unlabelled graph $H_i$ as follows:
	\begin{enumerate}
	\item Label the vertex $v$ with $c_v$ and label each $u_j$ with $c_{u_j}$.
	\item Label the remaining vertices with the same labels that their pre-images had in $\theta$, i.e., $L'(w) = L(h^{-1}(w))$.
	\end{enumerate}
\item Choose a receive transition $t_{u_i}$ for each $u_i$ s.t. $t_{u_i}$ is enabled at $c_{u_i}$ and is labelled by $??a$. Let $c_{u_i} \xrightarrow{t_{u_i}}
c_{u_i}'$.
\item Construct the labelled graph $H_{after}$ from the unlabelled graph $H_i$ as follows:
	\begin{enumerate}
	\item Label the vertex $v$ with $c_v'$ and label each $u_j$ with $c_{u_j}'$.
	\item Label the remaining vertices with the same labels that their pre-images had in $\theta$, i.e., $L'(w) = L(h^{-1}(w))$.
	\end{enumerate}
\item If $H_{after}$ lies in the set $S$, add $H_{before}$ as a basis element for the set $pre(S)$.
\end{enumerate}

It is clear that all the graphs that we are adding to our collection should be in $pre(S)$. We will now show that if $G$ is a graph in $pre(S)$, then there exists a graph $G' \le G$ which we would have added as a basis element to the set $pre(S)$ by the above procedure.

Let $G \in pre(S)$. Therefore, there should exist a transition from $G$ to some graph $F \in S$. Let this transition be obtained by broadcasting $!!a$ from the vertex $v \in G$ and which in turn is received by all its neighbors $u_1,u_2,\cdots,u_p \in G$. Let
the labels of $v$ in $G$ and $F$ be denoted by $c_v,$ $c'_v$ respectively and let the labels of each $u_i$ in $G$ and $F$ be denoted by $c_{u_i}$ and $c'_{u_i}$ respectively.
Since $F \in S$, there exists a basis element $F' \in \mathbb{B}$ 
s.t. $F' \le F$. Let $h$ be the required injection from $F'$ to $F$. 
We now have two cases:

\begin{itemize}
\item The node $v$ is in the image of $h$: Wlog let $u_1,u_2,\cdots,u_w$ be the neighbors of $v$ which are in the image of $h$. In this case, consider the first part of the above procedure in which we constructed labelled graphs out of $F'$. Since $v$ is in the image of $h$, consider the vertex $h^{-1}(v)$ and let its configuration in $F'$ be denoted by $c'_{h^{-1}(v)}$. Also, let the configurations of $h^{-1}(u_i)$ in $F'$ be denoted by $c'_{h^{-1}(u_i)}$. Since, there exists a broadcast transition from $c_v$ to $c'_v \ge c'_{h^{-1}(v)}$, it follows that the set $B_{!!a}^{h^{-1}(v)}$ is non-empty. Similar reasoning enables us to conclude that each of the sets $B_{??a}^{h^{-1}(u_i)}$ are also non-empty for each $i \le w$. Therefore, there exists configurations $c_{h^{-1}(v)} \le c_v$ from $B_{!!a}^{h^{-1}(v)}$ and
$c_{h^{-1}(u_i)} \le c_{u_i}$ from $B_{??a}^{h^{-1}(u_i)}$ which we would have picked during our procedure. Hence, the graph obtained out of $F'$ by replacing $c'_{h^{-1}(v)}$ with $c_{h^{-1}(v)}$ and $c'_{h^{-1}(u_i)}$
with $c_{h^{-1}(u_i)}$ gives us a graph $G' \le G$ which would have been constructed by our procedure. 

\item The node $v$ is not in the image of $h$: Wlog let 
$u_1,\cdots,u_w$ be the neighbors of $v$ which are in the image of $h$.
In this case, consider the graph $H_k$ which is the same as $F'$, except it contains one more vertex called $sp$ which is connected to exactly $u_1,\cdots,u_w$. Since there exists a broadcast transition $t$ labelled by $!!a$ from $c_v$, it follows that $c_t$ is non-empty and so
we can pick a configuration $c_{sp} \le c_v$ from $c_t$ which we use as a label for the vertex $sp$. 
Similar to the previous case, for the vertices $u_1,u_2,\cdots,u_w$, we can obtain configurations $c_{h^{-1}(u_i)} \le c_{u_i}$ from $B^{h^{-1}(u_i)}_{??a}$. Hence the graph obtained out of $H_k$ by replacing $c_v$ with $c_{sp}$ and $c_{u_i}$ with $c_{h^{-1}(u_i)}$ gives us a graph $G' \le G$.
\end{itemize}
\end{proof}

\begin{theorem}
Coverability in $k$-path bounded configurations is decidable.
\end{theorem}

\begin{proof}
Let $s$ be the given configuration. Consider the graph $G$ with only one vertex $v$ whose label is $s$. It is clear that the configuration $s$ can be covered iff the graph $G$ can be covered in the transition system $BP^k$ under the induced subgraph ordering. But by the previous lemmas, we have shown that $BP^k$ is a \textit{well-structured transition system} under the induced subgraph ordering with an effective pre-basis.
Therefore, coverability in $BP^k$ is decidable and this concludes the proof.
\end{proof}

Hence coverability in the broadcast semantics of $k$-path bounded topologies reduces to checking coverability in another WSTS!

\subsection{Clique topologies}

We prove a similar result for the set of all clique topologies.

\begin{lemma}
The set of all clique configurations forms a well- quasi ordering under the induced subgraph order.
\end{lemma}

\begin{proof}
Consider the poset $(P_f(S),\subseteq_{\le})$ where $P_f(S)$ is the set of all finite sub-multisets of $S$ and $S_1 \le S_2$ iff 
there exists an injection $h: S_1 \to S_2$ s.t. $s \le h(s), \forall \; s \in S_1$. It is well known that if $(S,\le)$ is a wqo,
then $(P_f(S),\subseteq_{\le})$ is also a wqo.

Let $G,G'$ be labelled clique configurations. It is clear that $L(G) \subseteq_{\le} L(G')$ iff $G \le G'$ under the induced subgraph ordering.
But $\subseteq_{\le}$ is a wqo. Hence, it follows that the set of all clique configurations forms a wqo under the induced subgraph ordering.
\end{proof}

The \textit{compatibility} property can be easily proved in an argument similar to the one given for $k$-path bounded graphs. The computation of pre-basis can be realized as follows: The algorithm given in the previous subsection, first selects a graph $G$ from the given basis $\mathbb{B}$ and then considers all $k$-path bounded graphs of size atmost $|G|+1$ which induce $G$ as a subgraph, after which it proceeds to construct a pre-basis from these $k$-path bounded graphs. We employ the same algorithm to construct a pre-basis for the clique topology as well, except in the first step, we replace the construction of the set of all $k$-path bounded graphs of size atmost $|G|+1$ which induce $G$, with the set of all clique graphs of size atmost $|G|+1$ which induce $G$. The proof of this algorithm follows from a similar proof given for the previous case. Hence we have,

\begin{theorem}
Coverability in clique configurations is decidable.
\end{theorem}

\subsection{Graphs with bounded diameter and degree}

It is known that the coverability problem for well-structured broadcast networks restricted to graphs of bounded diameter is undecidable, even when the underlying transition system is 
of finite state space \cite{clique}.
However we can regain decidability if along with bounded diameter, we also consider graphs of bounded degree.
To prove this, we use a non-trivial result of Hoffman and Singleton \cite{Moore}.
The result states for a fixed diameter $k$ and a degree $d$, the size of the largest (unlabelled) graph with diameter atmost $k$ and 
degree atmost $d$ is $M(k,d) = (k(k-1)^d-2)/(k-2)$. Hence, for finite state processes this immediately proves that the coverability problem is decidable.
But we can extend it in a straightforward way to the well-structured case as well. For a graph $G = (V,E)$ of bounded diameter and degree, 
consider the set $Lab(G,S) = \{G' \; | \; G' = (V,E,L) ; L: V \to S\}$, i.e., $Lab(G,S)$ is the set of all labelled graphs that can be obtained by
labelling the vertices in $G$ using labels from $S$.

\begin{lemma}
For a fixed graph $G$ of diameter $k$ and degree $d$, the set $Lab(G)$ is a well quasi ordering.
\end{lemma}

\begin{proof}
Suppose $G$ has $n$ vertices. Arbitrarily arrange the vertices in some order $v_1,\cdots,v_n$.
Notice then that each labelled graph $G' \in Lab(G,S)$ can be thought of as an element in $X^n$ where $G'$ is mapped to the $n$-tuple
$(L(v_1),L(v_2),\cdots,L(v_n))$. It is well known that if $(X,\le)$ is a wqo then $(X^n,\le)$ is also a wqo under the pairwise ordering.
From this the lemma immediately follows.
\end{proof}

For a fixed graph $G$, we can prove \textit{compatibility} and \textit{effective pre-basis} in a manner similar to the other cases.
Hence, we have

\begin{theorem}
Coverability problem for $k$-bounded diameter and $d$-bounded degree graphs is decidable.
\end{theorem}

\begin{proof}
Let $B = \{b_1,\cdots,b_m\}$ be a finite basis for the transition system $TS$ and let $s$ be the given configuration for which coverability needs to be determined.
Let $G$ be a fixed graph on $n$ vertices and consider the set $B_i = \{(b_{j_1},b_{j_2},\cdots,b_{j_{i-1}},s,b_{j_i},b_{j_{i+1}},\cdots,b_{j_n}): 
\text{ each } b_{j_l} \in B\}$ and let $B' = \cup_{1 \le i \le n} B_i$. Clearly the set $B'$ is finite. 

The above properties imply that given a fixed graph $G$ of bounded diameter and degree, it can be decided if any configuration from $B'$ can be covered
from $G$. But we know that the number of graphs with diameter $k$ and degree $d$ is finite. Hence, we can check if at least one configuration
from $B'$ can be covered from any of these graphs and so the coverability problem is decidable for bounded diameter and degree graphs.
\end{proof}

\section{Conclusion}

In this paper, we have defined broadcast networks for well-structured processes and proved decidability of coverability for 
various types of semantics. In particular, we have given an algorithm to determine if a given configuration can be covered in 
any run under the reconfiguration semantics. We have also studied decision procedures for various classes of restricted topologies 
which include the set of all path bounded graphs, the set of all cliques, and the set of all graphs with bounded diameter and degree.
A notable ingredient in these decision procedures is the construction of another well-structured transition system to decide 
coverability of configurations.

\nocite{*}
\bibliographystyle{eptcs}
\bibliography{ref}

\end{document}